\documentclass{llncs}
\usepackage{graphicx}
\usepackage{amssymb}
\usepackage{amsmath}
\usepackage{mathtools}
\usepackage[utf8]{inputenc}
\usepackage{nicefrac}
\usepackage{enumerate}
\usepackage[linesnumbered,vlined]{algorithm2e}
\usepackage{url}
\usepackage{caption}
\DeclareCaptionType{copyrightbox}
\usepackage{subcaption}
\usepackage{hyperref}
\usepackage[capitalize]{cleveref}
\usepackage{booktabs}

\graphicspath{ {./Pictures/} }

\DeclarePairedDelimiterX\Set[2]{\lbrace}{\rbrace}%
 { #1 \,\delimsize|\, #2 }
\renewcommand{\d}[1]{\ensuremath{\operatorname{d}\!{#1}}}
\newtheorem{fact}{Fact}
\crefname{fact}{Fact}{Facts}

\begin{document}
\title{Association Rule Mining using Maximum Entropy\thanks{This work is supported by the Danish National Research Foundation under the Sapere Aude program}}

\author{Rasmus Pagh \and Morten Stöckel}

\institute{IT University of Copenhagen \\ \email{\{pagh,mstc\}@itu.dk}}


\maketitle

\begin{abstract}
Recommendations based on behavioral data may be faced with ambiguous statistical evidence.
We consider the case of association rules, relevant e.g.~for query and product recommendations.
For example: Suppose that a customer belongs to categories A and B, each of which is known to have positive correlation with buying product C, how do we estimate the probability that she will buy product C?

For rare terms or products there may not be enough data to directly produce such an estimate --- perhaps we never directly observed a connection between A, B, and C.
What can we do when there is {\em no support\/} for estimating the probability by simply computing the observed frequency?
In particular, what is the right thing to do when A and B give rise to very different estimates of the probability of C?

We consider the use of {\em maximum entropy\/} probability estimates,
which give a principled way of extrapolating probabilities of events that do not even occur in the data set!
Focusing on the basic case of three variables, our main technical contributions are that (under mild assumptions):
1) There exists a simple, explicit formula that gives a good approximation of maximum entropy estimates, and
2) Maximum entropy estimates based on a small number of samples are provably tightly concentrated around the true maximum entropy frequency that arises if we let the number of samples go to infinity.

Our empirical work demonstrates the surprising precision of maximum entropy estimates, across a range of real-life transaction data sets. In particular we observe the average absolute error on maximum entropy estimates is a factor $3$--$14$ less compared to using independence or extrapolation estimates, when the data used to make the estimates has low support. We believe that the same principle can be used to synthesize probability estimates in many settings.
\end{abstract}

\newpage

\section{Introduction}\label{sec:int}

Recommender systems that try to assess probabilities, e.g.~for estimating probabilities based on the context of a particular user, may be faced with ambiguous statistical evidence.
For example, consider the task:
\emph{Customer is known to belong to categories A and B, each of which is known to increase the probability of buying product C by $50\%$, how do we estimate the probability that she will buy product C? Is it increased by $50\%$, $100\%$, or perhaps $125\%$?}

Of course we {\em may\/} have enough data on $S_1$, $S_2$, and $S_3$ to make this assessment by computing the observed probability.
But for rare terms or products there may not be enough data to directly produce such an estimate --- perhaps we never directly observed a connection between $S_1$, $S_2$, and $S_3$.
In the extreme case, what can we do when there is {\em no support\/} for estimating the probability by simply computing the observed frequency?
Most likely, even the number of observations of proper subsets of $S_1$, $S_2$, and $S_3$ will then be small enough that there is nonnegligible uncertainty about the pairwise correlations.

The difficulty of estimating probabilities of events occurring clearly depends on the distribution of the input, and on how much information we have about this distribution.
So rather than a classical approach that considers worst-case data, we should
consider ideas from statistical analysis.
%
The Maximum Entropy (maxent) Model is a method of statistical inference that based on partial knowledge of a distribution provides a \emph{maximum entropy estimate}.
Informally, it provides a probability prediction based on the distribution that has ``the least bias possible'' based on the given observations.
In this paper we consider the use of maximum entropy estimates in information retrieval contexts where estimates are sought of the probability that an item or term is of interest to a user.

\subsection{Motivating examples}

{\bf Movie recommendation.}
Suppose you know that a user loved ``The Rock'' and ``The Matrix''. What is the probability that he already saw and gave 5 stars to ``One Flew Over the Cuckoo's Nest''? We will try to answer this question using the smallest possible sample of the MovieLens data set, which is examined further in \Cref{sec:exp}. The difficulty is that only about 1 in 1000 people has seen all three movies and given them 5 stars. This means that to get a statistically significant answer (in absence of other information) we need to ask a very large set of people. Obviously, the less mainstream films you consider, the bigger this problem will become. See \Cref{fig:intro:ex2} for visualization of the setting.

However, it is considerably easier to obtain information on pairs of movies. About $2.5\%$ of people will have seen at least two of these movies and given them $5$ stars. This means that we can reliably estimate conditional probabilities based on significantly less data. Using the information that your user loved ``The Rock" gives a probability of about $11\%$ that he loves ``One Flew Over the Cuckoo's Nest". However, if we instead use the information that he loved ``The Matrix'' we get an estimate of about $7\%$. It seems that both these movies make it more likely that he will love ``One Flew Over the Cuckoo's Nest'', but how do we combine these pieces of information? It seems that anywhere in the range $11$-$18\%$ is a reasonable guess.

To resolve this ambiguity we again use a maximum entropy estimate based on subset frequencies. This estimate takes the correlation of ``The Rock'' and ``The Matrix'' into account, and arrives at an estimate of $14\%$ based on a data set of $673$ people in which nobody has given all three movies $5$ stars. When we consider a data set $100$ times larger it is possible to see how well this estimate fares: In the larger data set, $15\%$ of those who gave 5 stars to ``The Rock'' and ``The Matrix'' also gave 5 stars to ``One Flew Over the Cuckoo's Nest''. We generally find that maximum entropy estimates are surprisingly accurate across a wide range of data sets from different areas. 
\begin{figure}

        \centering
        \begin{subfigure}[b]{0.45\textwidth}
                \centering
                \includegraphics[width=\textwidth]{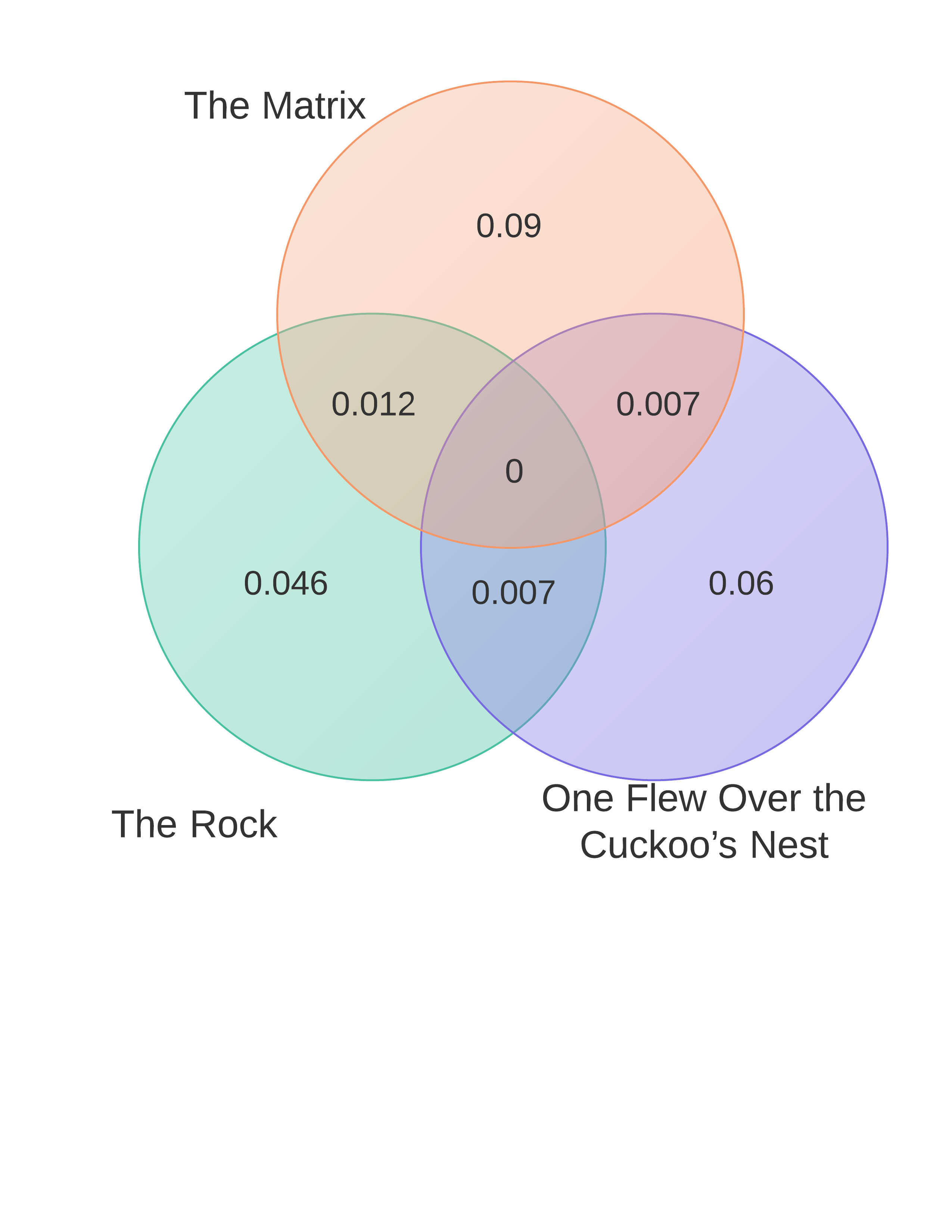}
                \caption{Venn diagram of distribution based on the sample.}
                \label{fig:intro:sample}
        \end{subfigure}%
        ~ 
        \begin{subfigure}[b]{0.45\textwidth}
                \centering
                \includegraphics[width=\textwidth]{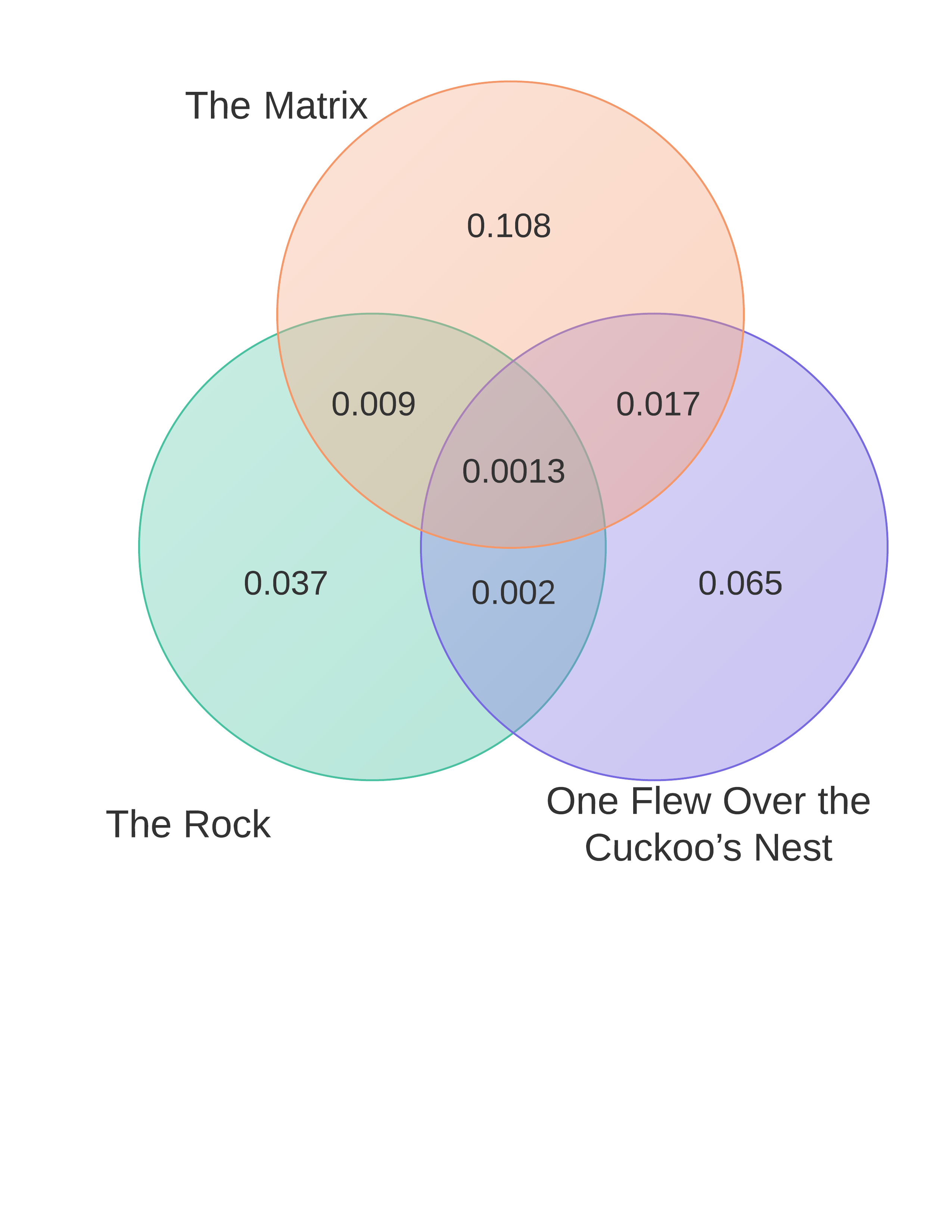}
                \caption{Venn diagram of distribution for the whole data set.}
                \label{fig:intro:full}
        \end{subfigure}
        \caption{Two distributions of movie watchers who love three selected movies. \Cref{fig:intro:sample} shows the observed probabilities in a sample of 1\% of the data set, from which we want to approximate $S_1 \cap S_2 \cap S_3$. This is consistent with $S_1$, $S_2$, and $S_3$ never occurring together. \Cref{fig:intro:full} shows the probabilities in the whole data set, and indeed it is not the case that loving two of the movies precludes loving the third one. Our findings are that a maximum entropy estimate in such a case is well-concentrated as opposed to independence or extrapolation estimators. In this example an independence assumption yields an estimate of $36$ occurrences, our maxent estimate yields $82$, while the true number of occurrences is $90$.}
        \label{fig:intro:ex2}
\end{figure}


{\bf Query completion.}
Consider the case where a search engine user types the words ``{\tt jordan} {\tt air}'' (followed by a space).
What words should be suggested to complete the query?

We consider the simple method of ignoring the order of words and relying on association rules, in our case obtained from a set of 2.1 million queries from a major US search engine.
Suppose we have two competing suggestions, ``{\tt force}'' and ``{\tt wholesale}'' (occurring in around 0.03\% and 0.06\% of queries, respectively).
Around 9\% of the queries that contain the word ``{\tt air}'' also contain the word ``{\tt force}''.
On the other hand, less than $0.00003\%$ of past queries contain ``{\tt jordan}'' {\em and\/} ``{\tt force}'', which means that the maximum entropy estimate for the probability of completing with ``{\tt force}'' becomes less than 1\%.
For comparison, both ``{\tt air}'' and ``{\tt jordan}'' significantly increase the probability that the word ``{\tt wholesale}'' occurs, to around 0.3\% and 1\%, respectively.
\Cref{fig:intro:ex1} summarizes the association rules involving two words.

If the combination of words had been just slightly more rare, we might have had no past queries containing them.
Thus, for ``long tail'' queries we need to rely on other methods for estimating the likelihood of a particular completion.

A maximum entropy estimate, or more precisely the approximation formula of~(\ref{eq:simplest}), shows that the user has a $5\%$ probability of completing with ``{\tt wholesale}''.
This is consistent with the data, which contains 31 queries including $\{${\tt jordan}, {\tt air}, {\tt wholesale}$\}$ out of a total of 575 queries containing ``{\tt jordan}'' and ``{\tt air}''.
\begin{figure}
        \centering
        \begin{tabular}{|c|l|}
			\hline
			{\bf Assoc.~rule} & {\bf Confidence}\\
			\hline
			{\tt air => force} & 0.091\\
			{\tt jordan => wholesale} & 0.010\\
			{\tt air => wholesale} & 0.0030\\
			{\tt jordan => force } & 0.00097\\
			\hline
        \end{tabular}
        \caption{ Association rules for the words {\tt force} and {\tt wholesale} in the  query set  {\tt jordan} and {\tt air}, respectively. The probabilities are sorted in decreasing order.
		In the whole data set, {\tt wholesale} occurs in a fraction 0.00064 of queries, and {\tt force} in a fraction 0.00028.
		By the first association rule, presence of the word {\tt air} increases the probability of the word {\tt force} hundreds of times, to over 9\%. However, a maximum entropy estimate correctly predicts less than 1\% probability of seeing {\tt force} when both {\tt air} and {\tt jordan} are present.
		On the other hand, the presence of {\tt air} and {\tt jordan} yield a maximum entropy probability for {\tt wholesale} of 3\%, close to the observed frequency of 5\%. }
        \label{fig:intro:ex1}
\end{figure}


We will argue by experimental evidence that in scenarios such as the one above, the maximum entropy estimate will
give a better prediction than both extrapolation and the independence model (see \Cref{sec:ests} for definitions), while still being efficiently computable.

\subsection{Our results}

{\bf Problem definition.} We consider the problem of estimating probabilities of conjunctions of boolean random variables, where each such conjunction occurs a statistically insignificant number of times in a data set of samples from the joint distribution (e.g., given by a complete data set). For some big data set $\mathcal{D}$ we consider a sample $D \subset \mathcal{D}$. Given such $D$ we wish to estimate event frequencies of $\mathcal{D}$ also in the difficult cases where the events do not occur in $D$.
In particular we will focus on triples: Let $I$ be the set of possible items, $|I|=n$, and $D$ be an $m \times n$ binary data set where each of the $m$ rows $D_i$ encodes a transaction $D_i \subseteq I$. For all singleton- and pair-subsets of $I$ we assume that we know the number of transactions which they occur in, i.e., all singleton and pair frequencies are known. For each $X \subset I$, $|X|=3$ where the frequency $\theta_X$ in $D$ is~$0$ we then wish to estimate $\theta_X$ in $\mathcal{D}$.

{\bf Our contribution.} We consider triple frequency estimation based on the principle of maximum entropy. 
Our main theoretical result is that a maximum entropy estimate based on a sample, which implies that the frequencies used as input to the estimator will have some relative error $\varepsilon$, will yield an estimate close to the true triple frequency under the maximum entropy assumption. We show this through a surprisingly simple estimator $\tilde{p}$ that approximates the maximum entropy estimate well when triple frequencies are small.
\begin{theorem}\label{thm:maxent:main}
Consider boolean random variables $X,Y,W$ where $\Pr(X)$, $\Pr(Y)$, $\Pr(W)$, $\Pr(XY)$, $\Pr(XW)$ and $\Pr(YW)$ are given.
Assume that the maximum entropy distribution consistent with these probabilities satisfies
\begin{align*}
\Pr(XYW)\leq \varepsilon \min(& \Pr(\overline{X}YW),\Pr(X\overline{Y}W),\Pr(XY\overline{W}),\Pr(\overline{XY}W),\\
& \Pr(\overline{X}Y\overline{W}),\Pr(X\overline{YW}),
\Pr(\overline{XYW})) \enspace .
\end{align*}
Then given probability estimates $p_{\cdot\cdot\cdot}$ such that
\begin{align*}
& \tfrac{p_{\overline{X}YW}}{\Pr(\overline{X}YW)},
\tfrac{p_{X\overline{Y}W}}{\Pr(X\overline{Y}W)},
\tfrac{p_{XY\overline{W}}}{\Pr(XY\overline{W})},
\tfrac{p_{\overline{XY}W}}{\Pr(\overline{XY}W)},
\tfrac{p_{\overline{X}Y\overline{W}}}{\Pr(\overline{X}Y\overline{W})},
\tfrac{p_{X\overline{YW}}}{\Pr(X\overline{YW})},
\tfrac{p_{\overline{XYW}}}{\Pr(\overline{XYW})}\\
& \in [(1-\varepsilon),(1+\varepsilon)],
\end{align*}
it holds that
\begin{align*}
\tilde{p} &= \frac{p_{XY\overline{W}} p_{X\overline{Y}W} p_{\overline{X}YW} p_{\overline{XYW}}}{ p_{X\overline{YW}}  p_{\overline{X}Y\overline{W}}  p_{\overline{XY}W} }\\
&\in \left[ (1- \mathcal{O}(\varepsilon))\Pr(XYW), (1+ \mathcal{O}(\varepsilon))\Pr(XYW)  \right]
\end{align*}
\end{theorem}
It follows from \Cref{thm:maxent:main} that a) using sampled data to perform maximum entropy estimates of probabilities in the bigger data is theoretically well-founded b) there is a simple explicit estimator, $\tilde{p}$, that approximates the maximum entropy estimate in the interesting case where the triple frequency is significantly smaller than the pair frequencies.

It is instructive to consider a less precise, but even simpler estimator for the case where, informally, there is no strong positive correlation among $X$, $Y$, and $W$, and $\Pr(\overline{XYW})$ is close to~1.
Then $p_{XY\overline{W}}/p_{p_{X\overline{YW}}} \approx p_{Y|X}$, the observed probability of $Y$ given $X$, and similarly $p_{X\overline{Y}W}/p_{p_{\overline{XY}W}} \approx p_{X|W}$ and $p_{\overline{X}YW}/p_{p_{\overline{X}Y\overline{W}}} \approx p_{W|Y}$, so we can approximate the triple frequency by:
\begin{align}\label{eq:simplest}
	p^* = p_{Y|X} p_{X|W} p_{W|Y}
\end{align}
\noindent
Applying (\ref{eq:simplest}) to estimate $\Pr(W|XY)$, we get the estimator \[p^{\#} = p_{Y|X} p_{X|W} p_{W|Y} / p_{XY} = p_{W|X} p_{W|Y} / p_W \text{; }\]that is, the factors by which conditioning on $X$ and $Y$ influence the probability of $W$ get multiplied.

{\bf Empirical study.} Our experimental evaluation on real data sets shows that maximum entropy estimates give meaningful, and often quite precise, frequency predictions also in cases where the independence estimate $\theta^i_X$ and the extrapolation estimate $\theta^e_X$ do not. The error of the estimator is well modeled by the assumption that transactions are independently sampled from a distribution having the estimated subset probabilities.

{\bf Overview.}
In \Cref{sec:prel} we introduce basic terms and notation followed by a description in \Cref{sec:ent} of how the maximum entropy estimate of an item set is computed. We then prove in \Cref{sec:noise} that the maxent estimate is not too sensitive to error on the input distribution. For the experimental evaluation we first show in \Cref{sec:exgen} that the maximum entropy estimate achieves better concentration in general and then in \Cref{sec:exlow} we discuss results on item sets of low statistically insignificant support.

\subsection{Related work}\label{sec:rel}

The principle of maximum entropy dates far back, but was introduced to information theory in a seminal work by E. T. Jaynes~\cite{1957PhRv..106..620J}. It has since seen applications in a large number of areas.

The maximum entropy distribution of $n$ random variables is known to be computable in time exponential in $n$ using the well-known Iterative Scaling algorithm~\cite{588021}. The running time is due to the fact that for $n$ variables, there are $2^n$ subsets of variables. In the general case, that is with no knowledge of the distribution, Tatti has proved
that querying the model is PP-hard \cite{Tatti:2006:CCQ:1146068.1711152}, which is (believed to be) harder than NP.

{\em Association rule mining\/}~\cite{Han:2005:DMC:1076797,Agrawal:1993:MAR:170036.170072} is a well-known and extensively studied problem, where a rule has the form $X \implies Y$ with $X,Y$ being disjoint subsets of random variables. In transactional data sets association rule mining traditionally relies on finding \emph{frequent itemsets}, i.e., for some set of items $I$ and a set of transactions $D$ over $I$ then one wishes to report back the sets $X \subseteq I$ that are contained in more than $s$ transactions, for a fixed threshold $s$.


The maxent distribution has been used as a model to measure how significant an itemset is, in the framework of frequent itemset mining, e.g. \cite{10.1109/ICDM.2007.43,Mampaey:2011:TMI:2020408.2020499}. The general approach is to compute the maximum entropy distribution (via the Iterative Scaling algorithm) and then compute the Kullback-Leibler divergence with the empirical distribution, from which a $p$-value can be found that is used to rank the item sets. Our approach is that we observe some sample of the subsets of the set of interest and then use these subsets to efficiently compute the maxent estimate.

For the case where all frequencies of the strict subsets are known, the maximum entropy model has been used by R. Meo~\cite{Meo:2000:TDV:363951.363956} by comparing the probability estimate under maximum entropy to the empirical probability in order to achieve a measure of ranking an itemset. One of the main open problems of \cite{Meo:2000:TDV:363951.363956} was determining the existence of a closed form formula for a maxent estimate given the subset probabilities. This was partially resolved in \cite{971763}, where the author provides a formula for the $1$-dimensional search space in which the maxent estimate lies, which is then traversed by binary search that is shown to converge to the maxent estimate.

{\bf Comparisons.} Our estimator takes as parameters all the single and pair-frequencies. The hypothesis present implicitly in our model is that data generally has weak third-order dependencies. This is also the reason Chow-Liu trees\cite{1054142}, which model only first and second-order dependencies, are known to be a good approximations of many observed distributions. The maximum entropy estimate used can be seen as an application of \cite{Meo:2000:TDV:363951.363956}, where singleton and pair frequencies are used to efficiently compute a maxent estimate in the interesting case where the estimand has no observed occurrences. One of the main open problems of \cite{Meo:2000:TDV:363951.363956} is an explicit formula for the maxent estimate - \Cref{thm:maxent:main} in this paper shows an explicit formula for an approximation of the maxent estimate under certain conditions.
As the $1D$ search of the maximum entropy estimate $\theta^m_X$ is determined when having all subset frequencies of $X$, we can compute a good maximum entropy estimate using a small constant number of iterations of binary search as opposed to computing the full maximum entropy distribution. We give a proof of this that is similar to that of Meo~\cite{971763}, with the distinction that where she shows that there exists constants such that the maxent estimate can be classified by setting particular equations to be equal to each other, in our proof the constants are explicitly stated.

\section{Frequency estimates of itemsets}\label{sec:mod}
\subsection{Preliminaries}\label{sec:prel}
We provide some definitions and notation that will be used throughout the paper.

Remember that a boolean random variable $A$ is a variable with values in $\{0,1\}$. A binary data set $D$ of observations of boolean random variables is
an $m \times n$ matrix consisting of $m$ binary $n$-sized vectors, index $0 \leq i \leq n-1$ corresponding to the outcome of boolean random variable $A_i$.
For a particular subset of the boolean variables $S \subseteq [n]$ a binary vector $\omega$ \emph{covers} $S$ iff $A_i = 1$ implies $\omega_i$ for every $i \in S$.
The \emph{frequency} $\theta_S$ of a set of boolean variables is the proportion of the $m$ row vectors in $D$ that covers $S$.\\
A \emph{distribution} $p$ over data $D$ is mapping \[ p: \{0,1\}^n \mapsto [0,1] \text{ s.t.} \sum_{\omega \in \{0,1\}^n} p(\omega) = 1\text{.} \]
For a distribution $p$ and a vector of $1$s $v$ we denote by $p(S = v) = p(S = 1) = \theta_S$ the probability $\Pr(\omega \text{ covers } S)$.

The \emph{empirical distribution} over data set $D$ is given as
\begin{equation}\label{eq:emp}
	q_D(a_1 = v_1, \ldots, a_n = v_n) =  \frac{ \left| \{ t \in D | t = v \}\right|}{m}
\end{equation}
We will denote by \emph{empirical frequency} the frequency according to the empirical distribution

Let a family of random variable sets $\mathcal{F}$ be \emph{satisfied} by a distribution $p: \{0,1\}^n \mapsto [0,1]$
iff for every $S \in \mathcal{F}$ it holds that $p( \omega \textrm{ covers } S ) = \theta_S$.

Finally, we say that a set of binary variables $X$ is \emph{downard closed} if for all strict subsets $S \subset X$ we have $\theta_S$, e.g., if we consider
the triple of items $s = \{I_1, I_2, I_3\}$ then $s$ is downward closed if we know the empirical singleton frequencies $\theta_{I_1}, \theta_{I_2}, \theta_{I_3}$ and empirical pair frequencies $\theta_{\{I_1,I_2\}}, \theta_{\{I_2,I_3\}}, \theta_{\{I_1,I_3\}}$.

We consider specifically the case where all itemsets of size $3$ (triples) are downward closed and the triples are the itemsets which we wish to estimate the frequency of.


\subsection{Estimation by extrapolation and independence assumption}\label{sec:ests}
We briefly describe the two estimators used for comparison. 

Let $X = \{I_1, I_2, I_3\}\subset I$, $|X|=3$, be the triple of interest from sampled data set $D \subset \mathcal{D}$. The independence model assumes the occurrences of random variables to not be correlated, thus we have:
\begin{equation*}
\theta^i_X = \theta_{I_1} \theta_{I_2} \theta_{I_3}
\end{equation*}
For the extrapolation estimator, let $occ(X)$ be the number of occurrences of $X$ in $D$. To estimate the frequency $\theta_X$ in $\mathcal{D}$ we have:
\begin{equation*}
\theta^e_X = occ(X)/|D|
\end{equation*}
Following from Chernoff bounds on independent variables $\theta^e_X$ is known to be a good estimate when $\theta_X$ is significant in $D$.

\subsection{Maximum entropy of itemsets}\label{sec:ent}
We are interested in estimating the frequency of a specific set of items $G \subseteq I$. We will estimate such a frequency using
the maximum entropy distribution, which intuitively can be thought of as the most uniform distribution given some
observed frequencies. We will compute the needed entries of the maximum entropy distribution to be used to compare to the empirical distribution.

For a family of itemsets $\mathcal{F}$ and a variable set $G$ let the projected family $\mathcal{F}_G$ be defined as \[ \mathcal{F}_G = \{ X \in \mathcal{F} \vert X \subset G, X \neq \emptyset\}.\]
Then letting $\mathcal{P}$ be the set of all possible probability distributions that satisfy $\mathcal{F}_G$,
the entropy of a distribution $p$ is given as
\begin{equation}
	H(p) =	- \sum_{ \omega \in \{0,1\}^{|G|}} p(\omega) \log p(\omega)
\end{equation}
The maximum entropy distribution $p*$ can be found by maximizing $H(p)$ over all $p \in \mathcal{P}$
\begin{equation}
	p^* = \arg\max_{p \in \mathcal{P}} H(p)
\end{equation}
We note that $|\mathcal{F}_G| = \mathcal{O}\left( 2^{|G|} \right)$ and if $\mathcal{F}_G = \emptyset$ then
 $p^*$ is the uniform distribution. The set $\mathcal{P}$ contains the empirical distribution (\Cref{eq:emp}) and hence is
non-empty by construction. We shall denote by \emph{maximum entropy estimate}  $\theta^m_X$ of an itemset $X$ the frequency of the itemset according to the maximum entropy distribution.

\subsubsection{Classifying and computing the maxent estimate}\label{sec:maxentcomp} For completeness we will state the approach used to compute the maximum entropy estimate of itemset frequency. We note that a similar proof of how to find the maxent estimate appears in \cite{971763}, however we give explicit constants in \Cref{eq:tmax} whereas their proof shows existence of the constants.

{\bf High-level description}. The overview of the proof is that for any $z>1$ variables, when $\Set{\theta_X}{X \subset G}$ is given for each variable then
the subspace of the joint probability space of $z$ random boolean variables is $1$-dimensional. It follows from this that the frequency estimate according to the maximum entropy distribution $p^*$ is located on a line segment. The estimate thus be computed by doing a simple binary search on this line segment and the main point of doing this is that we avoid computing the entire maximum entropy distribution.

Given $z$ random boolean variables and the marginal probabilities $\Set{\theta_X}{X \subset G}$ then the joint probability space is given by $2^z$ linear equations, where
the rank of the corresponding matrix is $2^z-1$. Let the $z$ boolean random variables $A_1, \ldots, A_z$ have probabilities $\Pr[A_1], \ldots, \Pr[A_z]$.
For $x \subseteq [z]$, $y = [z]\setminus x$ then denote by $f^=_{x}$ the probability $\Pr\left[\forall_{i \in x} A_i = 1 \wedge \forall_{j \in y} A_j = 0\right]$.
and by $\theta_x$ the probability $\Pr\left[\forall_{i \in x} A_i = 1\right]$.

\begin{lemma}\label{lem:marg}For random $z$ boolean variables $A_1, \ldots, A_z$ with feasible marginal probabilities $\Set{\theta_X}{X \subset [z]}$ then the space of feasible distributions is determined entirely by $f^{=}_{[z]}$. More generally,
for $S \subseteq [z]$ then $f^{=}_S$ follows from \Cref{eq:pincl}.
\begin{equation}\label{eq:pincl}
f^{=}_S = \sum_{S' \supseteq S} (-1)^{\left|S' \setminus S\right|} \theta_{S'}
\end{equation}
\end{lemma}
\begin{proof}
We prove this by induction on $|S|$. For $|S|=i$ we have \[ f^=_S = \sum_{S' \supseteq S} (-1)^{\left| S' \setminus S \right|} \theta_{S'} \]
For the inductive step we assume \Cref{eq:pincl} to hold for $|S| > i$.
Then for $|S| = i$, \Cref{eq:base} holds as the sum is over supersets of $S$.
By applying \Cref{eq:pincl} to the second term of the right hand side of \Cref{eq:base} we get
\Cref{eq:sub}.
\begin{align}
f^{=}_S &= \theta_S - \sum_{S' \supset S} f^=_{S'}\label{eq:base}\\
f^{=}_S &= \theta_S - \sum_{S' \supset S} \sum_{S'' \supseteq S'} (-1)^{\left|S'' \setminus S'\right|} \theta_{S''}\label{eq:sub}
\end{align}
The double sum of \Cref{eq:sub} can be split into two as shown in \Cref{eq:maxent:split}
\begin{equation}
f^{=}_S = \theta_S - \left( \sum_{S' \supseteq S} \sum_{S'' \supseteq S'} (-1)^{\left|S'' \setminus S'\right|} \theta_{S''} - \sum_{S' = S} \sum_{S'' \supseteq S'} (-1)^{\left|S'' \setminus S'\right|} \theta_{S''} \right)   \label{eq:maxent:split}
\end{equation}
The first double sum can be split into two parts \[ \sum_{S' \supset S} \sum_{S'' \supset S'} (-1)^{\left|S'' \setminus S'\right|} \theta_{S''} + \sum_{S' = S} \sum_{S'' = S'} (-1)^{\left|S'' \setminus S'\right|} \theta_{S''}\]
the first for which we will use \Cref{fa:sum} below.
\begin{fact}\label[fact]{fa:sum} For a set space $X$ and function $g: X \to \mathbb R$, for a double sum of sign-alternating supersets of any $x \in X$ we have \[ \sum_{x' \supset x} \sum_{x'' \supset x'} (-1)^{\left|x'' \setminus x'\right|} g(x) = 0 \] due to summands canceling out.
\end{fact}
Then by application of \Cref{fa:sum} to \Cref{eq:maxent:split} we get \Cref{eq:merge}, where the first double sum consists
only of the element $\theta_S$, hence we now arrive at the induction basis in \Cref{eq:basis}.
\begin{align}
f^{=}_S &= \theta_S - \left( \sum_{S' = S} \sum_{S'' = S'} (-1)^{\left|S'' \setminus S'\right|} \theta_{S''} - \sum_{S' = S} \sum_{S'' \supseteq S'} (-1)^{\left|S'' \setminus S'\right|} \theta_{S''} \right)\label{eq:merge}\\
f^{=}_S &= \theta_S - \left( \theta_S - \sum_{S' = S} \sum_{S'' \supseteq S'} (-1)^{\left|S'' \setminus S'\right|} \theta_{S''} \right) \notag\\
f^{=}_S &=  \sum_{S'' \supseteq S} (-1)^{\left|S'' \setminus S\right|} \theta_{S''} \label{eq:basis}
\end{align}
\end{proof}
Intuitively, \Cref{lem:marg} states that the space of probability distributions that satisfy
the marginal probabilities can be traversed by varying $f^{=}_{[z]}$. We note that \Cref{lem:marg} was used earlier by Calders \& Goethals~\cite[eq. (1)]{Calders:2007:NIM:1231311.1231319} and that we include the proof for sake of completeness.

We shall now argue that on this line through $2^z$ dimensional space, there is a unique point, i.e. a unique feasible distribution $p$,
that maximizes the entropy $H(p)$ and thus computing this point allows us to query the the maximum entropy distribution $p^*$ for a
$z$-sized set of variables.
Given a feasible distribution $x$ consisting of entries $x_S \geq 0$ for every $S \subseteq [z]$,
then by \Cref{lem:marg} the feasible distribution space $\mathcal{P}_f$ can be traversed by \Cref{eq:distspace}.
\begin{equation}\label{eq:distspace}
\mathcal{P}_f = x + t \cdot v, t \in (l,r)
\end{equation}
where $v$ is a $2^z$-sized $\{-1,+1\}$-vector with entries $v_S = (-1)^{z -|S|}$ for every $S \subseteq [z]$ and the range
$(l,r)$ is the range for which all coordinates in the vector $x + t \cdot v$ are non-negative.
Consider the partitioning of subsets $S\subseteq [z]$ into
$S_{\textup{even}} = \Set{S}{\left( z - |S| \right) \text{is even}}$  and $S_{\textup{odd}} = \Set{S}{\left( z - |S| \right) \text{is odd}}$.
Then the $t$-value borders are given below.
\begin{align*}
l &= \max \Set{x_S}{S \in S_{\textup{even}} }\\
r &= \min \Set{x_S}{S \in S_{\textup{odd}}}
\end{align*}
We note that feasible solution $x$ exists by construction since we have observed a feasible distribution and that vector $v$ corresponds to
the null space of the $2^z$-row matrix denoting the linear equalities which the joint distribution adheres to.

Location of the unique point corresponding to querying the max entropy distribution $p^*$ is shown in \Cref{lem:diff} below.
\begin{lemma}\label{lem:diff}
For a feasible distribution space $\mathcal{P}_f = x + t \cdot v, t \in (l,r)$ there exists a point $p_{max} \in \mathcal{P}_f$ s.t.
$ \forall p \in\mathcal{P}_f, p \neq p_{max} \implies H(p_{max}) \geq H(p)$.
In particular, \Cref{eq:diffres} below holds.
\begin{align}
\frac{\d{}}{\d{t}} H(x + t \cdot v) &= \frac{\d{}}{\d{t}} \left( \sum_{S \subseteq [z]} ( x_S + (-1)^{z - |S|}t) \log \left( \frac{1}{x_S + (-1)^{z-|S|} t}\right) \right) \label{eq:diffbase}\\
&= \sum_{S \subseteq [z]} (-1)^{z - |S|} \log \left(\frac{1}{x_S + (-1)^{z-|S|} t} \right)\label{eq:diffres}
\end{align}
\end{lemma}
\begin{proof}
We wish to show the existence of the scalar $t_{max}$ that optimizes the entropy,
i.e. $p_{max} = x + t_{max} \cdot v$.
\Cref{eq:diffres} follows form standard deriviate rules, from which we arrive at
\[\frac{\d{}}{\d{t}} H(x + t \cdot v) = \sum_{S \subseteq [z]} (-1)^{z - |S|}\left( \log \left(\frac{1}{x_S + (-1)^{z-|S|} t} \right) + \frac{\left( x_S + (-1)^{z-|S|} t \right)^2}{\left( x_S + (-1)^{z-|S|} t \right)^2} \right) \]
where the rightmost term will cancel out due to there being an equal number of odd- and even-sized subsets of $[z]$ for any integer $z$.
\end{proof}
\begin{corollary}\label{cor:diff}
The value $t_f \in (l,r)$ that maximizes $H(x + t \cdot v)$ while $x + t \cdot v$ is a feasible solution can be found as $t_f = \textup{median} \{l,r,t_{max} \}$.
\end{corollary}
\begin{proof}
From \Cref{lem:diff} we have that $t_{max}$ is the $t$-value that maximizes the entropy function $H(x + t \cdot v)$, which is at its unique maximum when \Cref{eq:tmax} holds.
\begin{equation}\label{eq:tmax}
 \sum_{S \in S_{\textup{odd}}} \log \left(\frac{1}{x_S - t} \right) = \sum_{S \in S_{\textup{even}}} \log \left(\frac{1}{x_S + t} \right)
\end{equation}
Let $t_f = \textup{median} \{l,r,t_{max} \}$. For the line segment spanned by end points $l$ and $r$, we have 3 cases: $t_{max} < l$ then $t_f = l$, $l \leq t_{max} \leq r$ then $t_f = t_{max}$ and $t_{max} > r$ then $t_f = r$.
The median of the values $l,r,t_{max}$ distinguishes these cases.
\end{proof}

The $t$ for which \Cref{eq:tmax} holds is $t_{\text{max}}$, that is, the equation holds strictly under maximum entropy. Since the solution space space of $(l,r)$ is always $1$-dimensional if all frequencies of the $2^z-1$ subsets are given, then we compute the value $t_{\text{max}}$ by binary search in the space. We note that this search is an approximation, but the binary search converges to values arbitrarily close to $t_{\text{max}}$ quickly in practice, e.g., $30$ iterations of binary search was used to produce the results in this paper. For an itemset $X$, the $t$-value we hold after $30$ such iterations in the search space is then the maximum entropy estimate $\theta^m_X$.

We note that by using a constant number of iterations for each binary search the computation takes time linear in the number of itemsets for which we wish to compute the maximum entropy estimate. Up to a constant factor this is equivalent to the extrapolation and independence assumption estimators.

\subsubsection{Maxent on noisy inputs}\label{sec:noise}
Recall that a view on data is that all data comes from some smaller sample $D$ sampled independently from a larger data set $\mathcal{D}$. As we wish to use the data from $D$ to reason about triples from $\mathcal{D}$ we will show that the error introduced by sampling does not hurt the maxent estimate too much, in particular we will show that a maxent estimate based on $D$ will not be far from a maxent estimate based on $\mathcal{D}$.
We will show that the maximum entropy estimate on a triple computed on input with relative error $0 \leq \varepsilon < 0.5$ is only a factor $\mathcal{O}(\varepsilon)$ from the maximum entropy estimate computed using the true distribution from $\mathcal{D}$ as input.

For a distribution $d$ over $3$ variables $X,Y,W$ we have \[ d = (\Pr(XYW), \Pr(XY \wedge \overline{W}), \ldots, \Pr(\overline{XYW}))\text{,}\] i.e., $|d| = 2^3$. Let $|d_i|$ be the number of non-negated literals in $d_i$. The entries of $d$ can (as in \Cref{sec:maxentcomp}) be partitioned into two sets, $u = u_1, \ldots, u_4$ and $ l = l_1, \ldots, l_4$ where $d_i \in u$ if $3-|d_i|$ is odd and $d_i \in l$ otherwise.

Let two polynomials $L_{\varepsilon_1}(t)$ and $U_{\varepsilon_1}(t)$ with error $0 \leq |\varepsilon_1| < 1$ be defined
\begin{align}
L_{\varepsilon_1}(t) & = (l_1 + \varepsilon_1 + t)(l_2 + \varepsilon_1 + t)(l_3 + \varepsilon_1 + t)t \label{eq:upol} \\
U_{\varepsilon_1}(t) & = (u_1 + \varepsilon_1 - t)(u_2 + \varepsilon_1 - t)(u_3 + \varepsilon_1 - t)(u_4 + \varepsilon_1 - t) \label{eq:lpol}
\end{align}
where $t = \Pr(XYW)$ is the triple frequency. When we have
\begin{equation}
L_0(t) = U_0(t)
\end{equation}
then $t$ is the triple frequency under maximum entropy. We seek to bound the error on the output caused by the input error $\varepsilon_1$. Letting $t_0$ be the solution to $L_0(t) = U_0(t)$ and $t_{\varepsilon_1}$ be the solution to $L_{\varepsilon_1}(t) = U_{\varepsilon_1}(t)$ we wish to bound the error on $t_{\varepsilon_1}$ in terms $t_0$.
We will show the following lemma.
\begin{lemma}\label{lem:bounds}
Let two polynomials with additive error $t$ on the terms and where $0 < l_i, u_i, \leq 1$ be defined as below.
\begin{align}
\tilde{L}(t) &= (l_1+t)(l_2+t)(l_3+t)\\
\tilde{U}(t) &= (u_1-t)(u_2-t)(u_3-t)(u_4-t)
\end{align}
For all $t \in \left[0; \varepsilon \min{u_1, \ldots, u_4, l_1, \ldots, l_3}\right]$, where $ 0 < \varepsilon \leq 1$, 
the following bounds on \Cref{eq:upol,eq:lpol} hold 
\begin{align}
\tilde{L}(0)t &\leq L_0(t) \leq (1+\varepsilon)^3 \tilde{L}(0)t \label{eq:uin}\\
(1-\varepsilon)^4 U_0(0) &\leq \tilde{U}(t) \leq U_0(0) \label{eq:lin}
\end{align}
\end{lemma}
\begin{proof}
The first inequality of \cref{eq:uin} follows trivially from $t > 0$ and from $\tilde{L}$ being monotonically increasing in $t$. For the second inequality let $l_{\min} = \min{l_1,\ldots,l_3}$ and $t \leq \varepsilon l_{\min}$, $0 < \varepsilon \leq 1$. It follows that
\begin{align*}
L_0( t) & \leq (l_1+\varepsilon l_{\min})(l_2+ \varepsilon l_{\min})(l_3+ \varepsilon l_{\min})\varepsilon l_{\min}\\
& \leq (1+\varepsilon)^3 l_1 l_2 l_3 t\\
&= (1+\varepsilon)^3 \tilde{L}(0) t
\end{align*}
The first inequality of \cref{eq:lin} follows analogously; let $u_{\min} = \min{u_1,\ldots,u_4}$ and $t\leq \varepsilon u_{\min}$, $0 < \varepsilon \leq 1$. We then have
\begin{align*}
\tilde{U}(t) &\geq (u_1-\varepsilon u_{\min})(u_2-\varepsilon u_{\min})(u_3-\varepsilon u_{\min})(u_4-\varepsilon u_{\min})\\
 &\geq (1-\varepsilon)^4 U_0(0)
\end{align*}
The second inequality again follows from $t>0$ and $\tilde{U}$ being monotonically decreasing in $t$.
\end{proof}
It follows that a simple approximate formula for maximum entropy estimates on triples exist.
\begin{lemma}\label{lem:tint}
For a triple with distribution over entries $\mathcal{D} = l_1, \ldots, l_3, u_1, \ldots, u_4$ let
\begin{equation*}
\tilde{t} = \frac{u_1 u_2 u_3 u_4}{l_1 l_2 l_3}
\end{equation*}
The triple frequency under maximum entropy $t < \varepsilon \min{\mathcal{d}}$ for $0 < \varepsilon \leq 0.5$ can be bounded in terms of $\tilde{t}$
\begin{equation}\label{eq:tbound}
t \in \left[ (1-6.5\varepsilon) \tilde{t},  \tilde{t}  \right]
\end{equation}
\end{lemma}
\begin{proof}
Recall that $t$ takes value under maximum entropy when $U_0(t)/L_0(t) = 1$.
By \Cref{lem:bounds} we have a lower bound on $U_0(t)/L_0(t)$
\begin{align*}
(1-\varepsilon)^4 U_0(0) / (1+\delta)^3 \tilde{L}(0) &\leq U(t)_0/L_0(t)\\
(1-\varepsilon)^4 u_1 u_2 u_3 / (1+\varepsilon)^3 l_1 l_2 l_3 &\leq U_0(t)/L_0(t) \\
\frac{(1-\varepsilon)^4}{(1+\varepsilon)^3} \tilde{t} & \leq U_0(t)/L_0(t)
\end{align*}
By expansion and using $0 < \varepsilon \leq 0.5$ we get the needed lower bound of \cref{eq:tbound}.
\begin{align*}
1 - c*\varepsilon \geq \frac{(1-\varepsilon)^4}{(1+\varepsilon)^3} \\
c \geq -\varepsilon^3 + 5\varepsilon^2 - 3\varepsilon + 7\\
c \geq 6.5
\end{align*}
Equivalently we have the needed upper bound
\begin{align*}
U_0(t)/L_0(t) &\leq U_0(0) / \tilde{L}(0)\\
U_0(t)/L_(t) &\leq u_1 u_2 u_3 / l_1 l_2 l_3\\
U_0(t)/L_(t) &\leq \tilde{t}
\end{align*}
\end{proof}
We will now assume relative error on the distribution entries $l_i, u_i$. The following lemma holds analogously to \Cref{lem:tint}.
\begin{lemma}\label{lem:appr}
Let distribution $d = (l_1, \ldots l_3, u_1, \ldots, u_4)$. Let an approximation of distribution $d$ be defined by $u'_i \in [(1-\varepsilon)u_i, (1+\varepsilon)u_i]$ and  $l'_i \in [(1-\varepsilon)l_i, (1+\varepsilon)l_i]$ for each $u_i, l_i \in d$.
\begin{align*}
\tilde{t}_{r} &= \frac{u'_1 u'_2 u'_3 u'_4}{l'_1 l'_2 l'_3}\\
\tilde{t} &= \frac{u_1 u_2 u_3 u_4}{l_1 l_2 l_3}
\end{align*}
It holds that
\begin{equation}\label{eq:tapprox}
\tilde{t}_{r} \in \left[ (1-6.5\varepsilon') \tilde{t},  \tilde{t}  \right]
\end{equation}
\end{lemma}
\begin{proof}(\Cref{thm:maxent:main})
The theorem follows directly from \Cref{lem:tint,lem:appr}.
\end{proof}

\section{Experimental Results}\label{sec:exp}
Our experiments were conducted on $5$ real datasets shown in \Cref{tab:data}. For each dataset we prune the singletons with support below a specified threshold. We perform this pruning in order to construct datasets where intuitively the independence model has a chance to do well as it relies solely on the singletons, but also to keep the number of items down to a practical level, as the time complexity with $n$ items is $\mathcal{O}(n^3)$ per experiment.
AOL Queries \footnote{\url{http://www.gregsadetsky.com/aol-data/}} is a uniform sample of the infamous AOL seach terms dataset. Docwords is transactions of words occurring together in documents. From MovieLens \footnote{\url{http://www.grouplens.org/node/73}} we create a dataset where a transaction is the set of movies which a particular user rated 5/5 stars. Retail \footnote{\url{http://fimi.ua.ac.be/data/}} is shopping baskets from an anonymous Belgian supermarket.

{\bf Overview of experiments.} We start by showing how the independence and maxent estimators perform on the full datasets, i.e., when there is significant support for the triple whose frequency is to be estimated as well as its subsets. The maxent estimator shows better concentration and less variance even for the low-support triples. We then perform experiments for all datasets where $1\%$ of the transactions are sampled and we wish to estimate the frequency of triples in the whole data set. For every triple $X$ with $30 \leq occ(X) \leq 100$ we use the three estimators $\theta^m_X$, $\theta^i_X$ and $\theta^e_X$. As $X$ occurs at most $100$ times in the full dataset, it occurs in expectation at most once in the sample.

We also study precision and recall for the problem of approximating the set of frequent triples based on the estimators. Again we consider two cases: 1) Estimates are based on the full dataset, where the set to approximate consists of the $10\%$ highest frequencies among all triples, and 2) as in the first case, but with estimates based on a $1\%$ sample. In the latter case, if the  threshold for being in top 10\% is $\Delta$, then we include in the estimate all triples that are estimated to have at least $0.9 \Delta$ occurrences. The number 0.9 was experimentally found to yield a good precision/recall tradeoff for the maximum entropy estimates.

Finally, using again a $1\%$ sample of the data, we compute the average ratio between the error made by our maximum entropy estimate and estimates made by independence and extrapolation, respectively.

In summary our experiments show:
\begin{enumerate}
\item In most cases, the maximum entropy estimator provides the best estimate for low-support triples.
\item Sampling with higher probability increases concentration greatly for $\theta^m_X$ even when the sample is still of insufficient size for $\theta^e_X$ to be useful.
\item In almost all cases studied, the precision and recall are strictly better for $\theta^e_X$ (see \cref{tab:prfull,tab:prsample}).
\item When predicting the occurences of triples in the full dataset using a $1\%$ sample, we show $\theta^m_X$  improves the absolute error by a factor $3$ -- $14$ compared to $\theta^i_X$ and $\theta^e_X$ (see \Cref{tab:estratio}).
\end{enumerate}
\begin{table}[!t]
\centering
\begin{tabular}{|c|ccc|}
\hline
Name & Occ. Threshold& \#Items & \#Transactions\\
\hline
AOL Queries & $500$ & $211$ & $144038$\\
Docwords & $3000$ & $142$ & $49078$\\
Movielens & $3000$ & $85$ & $67312$\\
Retail & $500$ & $85$ & $88162$\\
\hline
\end{tabular}
\caption{Datasets used. All datasets are from real data and have previously been used for data mining purposes.}\label{tab:data}
\end{table}

\begin{table}[t!]
\centering
\begin{tabular}{|l|llll|}
\hline
 & \multicolumn{2}{ c }{Independence} & \multicolumn{2}{ c| }{Maxent}\\
Dataset & Precision & Recall & Precision & Recall\\
\hline
AOL Queries & $0.0$ & $0.0$ & $0.54$ & $1.0$\\
Docwords & $0.93$ & $0.43$ & $0.92$ & $0.93$\\
MovieLens & $1.00$ & $0.003$ & $0.80$ & $0.96$\\
Retail & $1.00$ & $0.43$ & $0.99$ & $0.97$\\
\hline
\end{tabular}
\caption{Precision-recall for full datasets. Relevant triple threshold $\Delta$ is such that relevant triples are among the $10\%$ most frequent.}\label{tab:prfull}
\end{table}

\begin{table}
\centering
\begin{tabular}{|l|llllll|}
\hline
 & \multicolumn{2}{ c }{Ind.} & \multicolumn{2}{ c }{Maxent} & \multicolumn{2}{ c| }{Extrapol.} \\
Dataset & Prec & Rec & Prec & Rec & Prec & Rec\\
\hline
AOL Queries & $1.0$ & $0.001$ & $0.23$ & $0.89$ & $0.31$ & $0.67$ \\
Docwords & $0.88$ & $0.44$ & $0.56$ & $0.72$ & $0.53$ & $0.58$ \\
MovieLens & $1.0$ & $0.009$ & $0.51$ & $0.90$ & $0.49$ & $0.85$ \\
Retail & $0.93$ & $0.40$ & $0.51$ & $0.78$ & $0.49$ & $0.73$ \\
\hline
\end{tabular}
\caption{Precision-recall for sampled datasets.  Relevant triple threshold $\Delta$ is such that relevant triples are among the $10\%$ most frequent. Triples are reported when $occ \geq 0.9\Delta$ as this was observed to maximize precision/recall for $\theta_X^e$.}\label{tab:prsample}
\end{table}

\begin{table}
\centering
\begin{tabular}{|l|ll|}
\hline
Dataset & Independence & Extrapolation \\
\hline
AOL Queries & $7.91$ & $7.58$ \\
Docwords & $6.31$ & $14.32$ \\
MovieLens & $11.55$ & $7.22$ \\
Retail & $3.22$ & $4.42$ \\
\hline
\end{tabular}
\caption{For $n$ estimates and two estimators $est_1$ and $est_2$ the table shows the normalized absolute error ratio: $\frac{1}{n} \sum_X \left(|est_1(X) - occ(X)| / |est_2(X) - occ(X)|\right)$, where $est_2$ is maximum entropy and $est_1$ is independence and extrapolation respectively for the two columns.}
\label{tab:estratio}
\end{table}

\subsection{Maxent vs. independence for full datasets}\label{sec:exgen}
For all datasets we ran the estimators on every triple with occ $>30$ on the full (unsampled) dataset. The figures show the concentration of the estimates by for each triple plotting its observed value (Y-axis) and estimated value (X-axis). The plots are shown in \Cref{fig:full:1,fig:aol:all}.

For the Docwords dataset both the maxent estimate (\Cref{fig:doc:ent}) and the independence estimate (\Cref{fig:doc:ind}) are concentrated around the $X=Y$, which denote the line of optimal estimations while for MovieLens we observe similar concentration for maxent (\Cref{fig:mov:ent}) while the independence estimator underestimates slightly (\Cref{fig:mov:ind}) and Retail behaves equivalently in this setting. For AOL Queries (\Cref{fig:aol:all}) using independence estimates we observe similar great under estimation due to high positive correlation, while the maxent estimator overestimates slightly but is far more concentrated.

Our precision and recall computations (\Cref {tab:prfull}) is based on relevance threshold $\Delta$ being such that relevant triples are among the top $10\%$ most frequent and we report a triple if the estimate is at least $\Delta$. Note that the high precision for the independence estimate is due to high underestimation - the estimators report too few triples as relevant, as the recall shows. An extreme case of this is the AOL dataset that reports zero triples. Maxent has similar precision but much higher recall for all datasets.
\begin{figure*}
        \centering
        \begin{subfigure}[b]{0.5\textwidth}
                \centering
                \includegraphics[width=\textwidth]{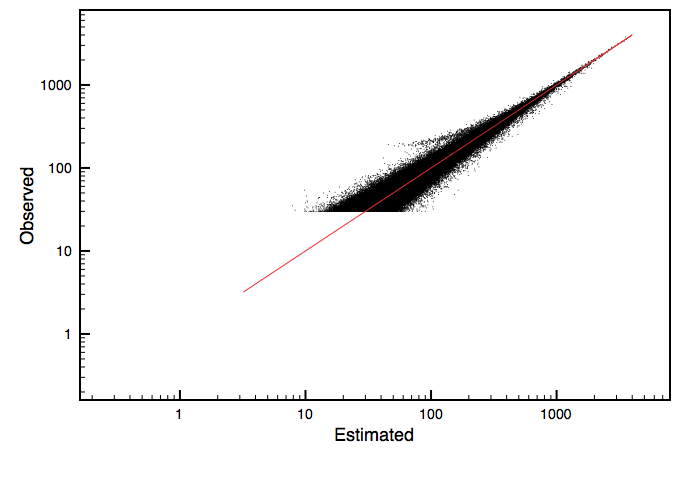}
                \caption{Max. entropy estimates for all triples of occ $>30$ in Docwords.}
                \label{fig:doc:ent}
        \end{subfigure}%
        ~ 
        \begin{subfigure}[b]{0.5\textwidth}
                \centering
                \includegraphics[width=\textwidth]{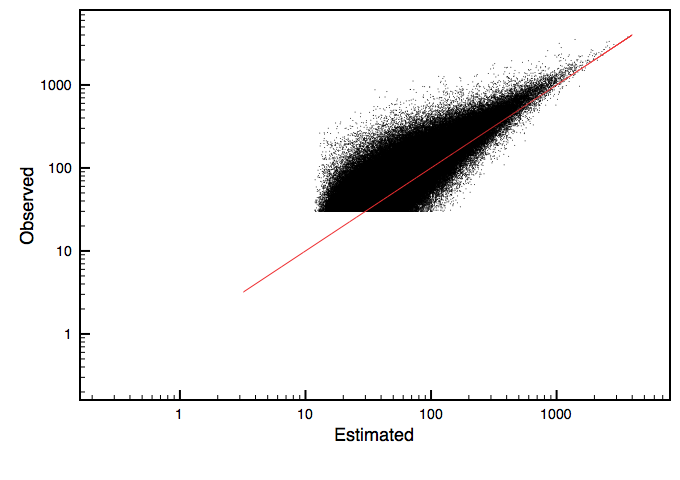}
                \caption{Independence estimates for all triples of occ $>30$ in Docwords.}
                \label{fig:doc:ind}
        \end{subfigure}

        \begin{subfigure}[b]{0.5\textwidth}
                \centering
                \includegraphics[width=\textwidth]{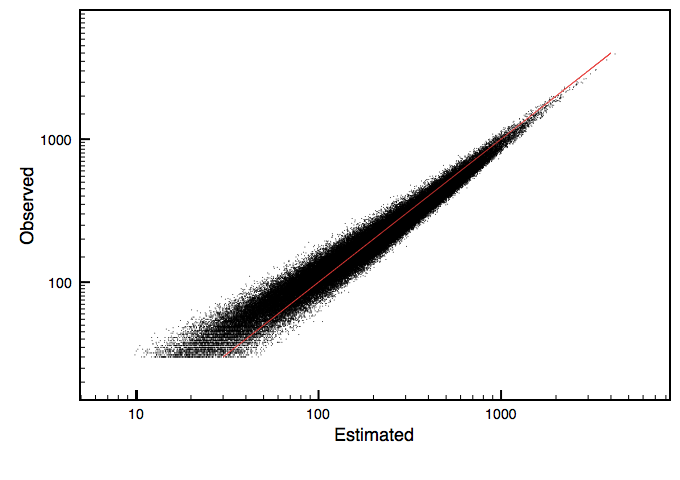}
                \caption{Max. entropy estimates for all triples of occ $>30$ in Movielens.}
                \label{fig:mov:ent}
        \end{subfigure}%
        ~ 
        \begin{subfigure}[b]{0.5\textwidth}
                \centering
                \includegraphics[width=\textwidth]{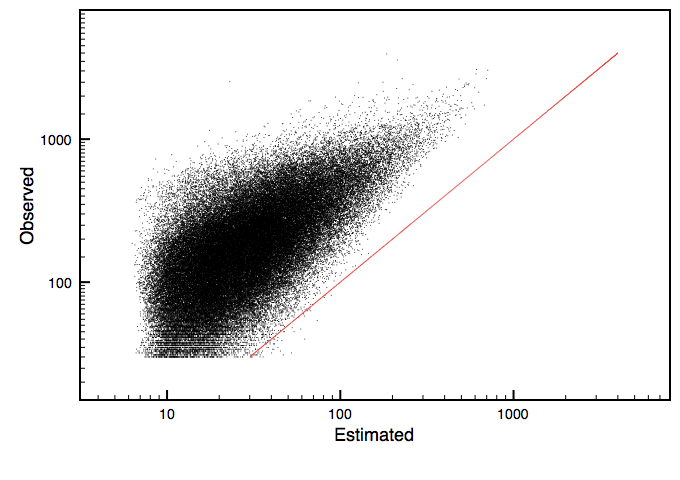}
                \caption{Independence estimates for all triples of occ $>30$ in Movielens.}
                \label{fig:mov:ind}
        \end{subfigure}
        \caption{Concentration plots for Docwords and Movielens datasets. Each point is a triple, the Y-value of a point is the empirical number of occurrences (occ) of the triple while the X-value is the estimated number of occurences, using either independence or maxent estimators. The red line is X=Y. We observe better concentration on the maxent estimates, in particular when the statistical significance is high.} \label{fig:full:1}
\end{figure*}

%

\begin{figure*}
        \centering
        \begin{subfigure}[b]{0.5\textwidth}
                \centering
                \includegraphics[width=\textwidth]{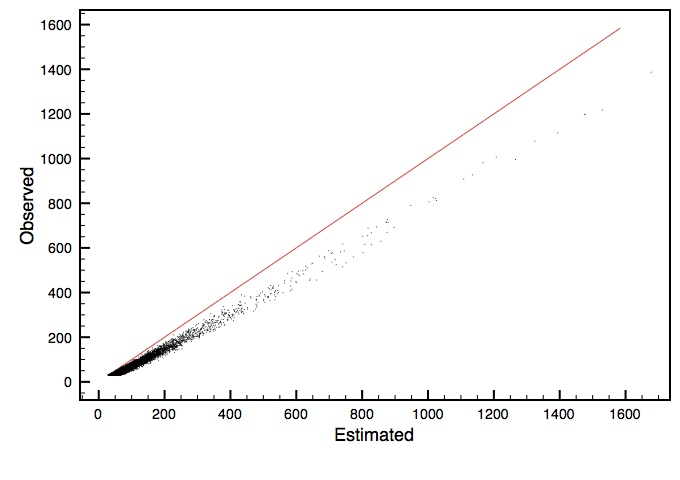}
                \caption{Maxent estimates for all triples of occ $>30$ in AOL Queries.}
                \label{fig:aol:maxent}
        \end{subfigure}%
        ~ 
        \begin{subfigure}[b]{0.5\textwidth}
                \centering
                \includegraphics[width=\textwidth]{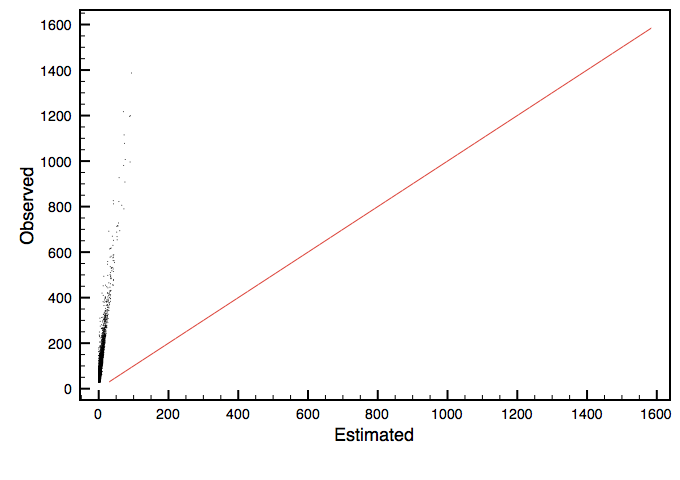}
                \caption{Independence estimates for all triples of occ $>30$ in AOL Queries.}
                \label{fig:aol:ind}
        \end{subfigure}
        \caption{Concentration plots for AOL Queries. Each point is a triple, the Y-value of a point is the empirical number of occurrences (occ) of the triple while the X-value is the estimated number of occurences, using either independence or maxent estimators. The red line is X=Y. We observe better concentration for Maxent for all occurences. \label{fig:aol:all}}
\end{figure*}

\subsection{Low-support itemsets}\label{sec:exlow}
On the same datasets we now examine triples $X$ where $30 \leq occ(X) \leq 100$ and we sample independently at random every transaction with probability $1/100$. We wish to estimate the $\theta_X$ in the full dataset, but since we have $occ(X) \leq 100$ then following from independent sampling the expected number of occurrences of $X$ in our sample is $\leq 1$. We will perform estimates using the extrapolation estimate $\theta^e_X$, the independence estimate $\theta^i_X$ and the maxent estimate $\theta^m_X$. We restrict ourselves to triples where all pairs occur in the sample.

The extrapolation estimator $\theta^e_X$ performs similarly on all datasets. While $\theta^e_X$ is an unbiased estimator of $\theta_X$, the variance is large as a consequence of the sampling - by the mode of independent Bernoulli trials we have that the most likely outcome of a triple is $\lceil \mu \rceil = 1$, with the outcome $0$ and $2$ slightly less probable. On the AOL, Retail and MovieLens datasets our experiments conclude similarly: $\theta_X^e$ doesn't give meaningful estimates, e.g., zero for the unsampled triples and overestimates on the sampled triples, while $\theta_X^i$ underestimates greatly and $\theta_X^m$ is fairly concentrated. See \Cref{fig:aol:ratio} for example of all three estimators on AOL. For Docwords in this setting, we get better concentration for $\theta_X^i$ than $\theta_X^m$ - this can be explained by pair frequencies being more vulnerable to noise introduced by sampling than single frequencies and $\theta_X^i$ being well concentrated for Docwords (recall \Cref{fig:doc:ind}). However, we observe that if we raise the sampling probability to $1/20$ from $1/100$ then $\theta_X^m$ has better concentration. In summary, there is a sampling rate where $\theta_X^e$ is very poorly concentrated where $\theta_X^m$ outperforms $\theta_X^i$ on all datasets.

Precision and recall values (\Cref{tab:prsample}) were computed by setting the relevance threshold $\Delta$ to be s.t. if $occ(X) \geq \Delta$ then triple $X$ is among the $10\%$ most frequent triples in the entire dataset with $occ(X) \geq 30$. Triples are reported if the occurrence estimate, which is based on the $1/100$ independent sample of all transactions, is at least $0.9\Delta$ as this value was observed experimentally to yield a good tradeoff. AOL and MovieLens using $\theta_X^i$ has full precision due to reporting very few triples. Note that $\theta_X^e$ has an advantage in terms of precision/recall due to it mostly performing overestimates, i.e., if a triple occurs in the sample then it will likely be reported -- even so we see $\theta_X^m$ being better than $\theta_X^e$ except for AOL where only recall is better.

In the same setting, i.e., triples where $occ(X) \geq 30$ and using an independent $1\%$ sample we compute the normalized ratio of the absolute error between estimators $\theta_X^i$, $\theta_X^e$ and $\theta_X^m$. That is, letting $est_m(X),est_i(X)$ denote estimates of triple $X$ by maxent and independence respectively, the normalized ratio is $\frac{1}{n} \sum_X \frac{|est_i(X) - occ(X)|}{|est_m(X) - occ(X)|}$. This experiments show (see \Cref{tab:estratio}) that for all datasets in this setting we would improve, on average, our absolute error by a factor of $3$ -- $14$ by switching from independence or extrapolation to maximum entropy.

 \begin{figure}
         \centering
          \includegraphics[width=0.5\textwidth]{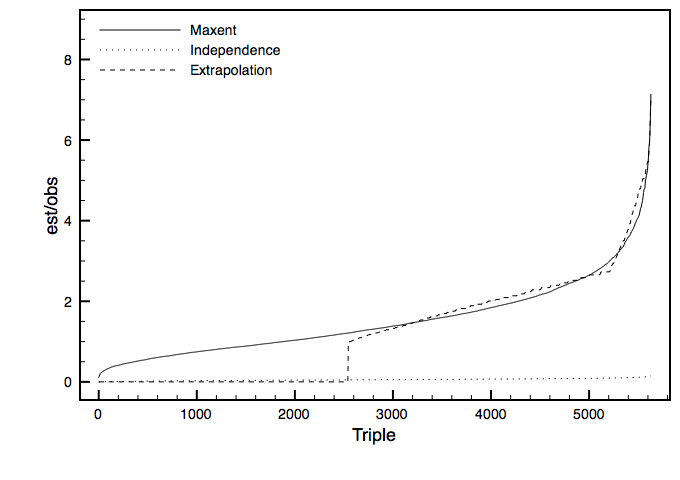}
         \caption{est/obs distribution plots for AOL Queries sampled at $1/100$. We observe that independence underestimates greatly, while maxent has better concentration than extrapolation, in particular for the low occurence triples.  \label{fig:aol:ratio}}
 \end{figure}

\bibliographystyle{plain}
\bibliography{maxent}
\end{document}